\documentclass[a4paper]{article}
\usepackage[T1]{fontenc}
\usepackage{authblk}
\usepackage{graphicx}
\usepackage{amsthm,amsmath,cases}
\usepackage{algorithm,algorithmic}
\usepackage[bookmarksdepth=4]{hyperref}

\newtheorem{lemma}{Lemma}
\newtheorem{corollary}{Corollary}
\newtheorem{theorem}{Theorem}

\newcommand{\lpsuf}{\mathit{lpsuf}}
\newcommand{\lppref}{\mathit{lppref}}
\newcommand{\dpal}{d'}
\newcommand{\bp}{\mathit{bp}}
\newcommand{\first}{\mathit{first}}
\newcommand{\second}{\mathit{second}}
\newcommand{\prefPal}{\mathit{prefPal}}
\newcommand{\ismups}{\mathsf{ismups}}
\newcommand{\zeronode}{\mathtt{0}\mbox{-node}}
\newcommand{\minusonenode}{\mathtt{-1}\mbox{-node}}
\newcommand{\lpp}{\mathit{lpp}}
\newcommand{\lps}{\mathit{lps}}
\newcommand{\len}{\mathit{len}}
\newcommand{\BegPair}{\mathit{BegPair}}
\newcommand{\sufLink}{\mathit{slink}}
\newcommand{\inSufLink}{\mathit{inSL}}
\newcommand{\rightmost}{\mathit{rm}}
\newcommand{\secondrightmost}{\mathit{srm}}
\newcommand{\node}{\mathit{node}}
\newcommand{\updatebp}{\mathit{update\_bp}}
\newcommand{\DPal}{\mathsf{DPal}}
\newcommand{\MUPS}{\mathsf{MUPS}}
\newcommand{\eertree}{\mathsf{eertree}}
\DeclareMathOperator{\pal}{\mathit{pal}}
\begin{document}
\title{Palindromic Trees for a Sliding Window \\ and Its Applications}
\author[1,2]{Takuya~Mieno}
\author[1]{Kiichi~Watanabe}
\author[1]{Yuto~Nakashima}
\author[1,3]{Shunsuke~Inenaga}
\author[4]{Hideo~Bannai}
\author[1]{Masayuki~Takeda}
\affil[1]{Department of Informatics, Kyushu University.

\texttt{\{takuya.mieno,kiichi.watanabe,
yuto.nakashima,inenaga,takeda\}@inf.kyushu-u.ac.jp}}
\affil[2]{Japan Society for the Promotion of Science.}
\affil[3]{PRESTO, Japan Science and Technology Agency,}
\affil[4]{M\&D Data Science Center, Tokyo Medical and Dental University.

\texttt{hdbn.dsc@tmd.ac.jp}}
\date{}
\maketitle
\begin{abstract}
  The palindromic tree~(a.k.a. eertree) for a string $S$ of length $n$ is a tree-like data structure
  that represents the set of all distinct palindromic substrings of $S$, using $O(n)$ space~[Rubinchik and Shur, 2018].
  It is known that, when $S$ is over an alphabet of size $\sigma$ and is given in an online manner,
  then the palindromic tree of $S$ can be constructed in $O(n\log\sigma)$ time with $O(n)$ space.
  In this paper, we consider the sliding window version of the problem:
  For a sliding window of length at most $d$,
  we present two versions of an algorithm which maintains the palindromic tree of size $O(d)$ for every sliding window $S[i..j]$ over $S$,
  where $1 \leq j-i+1 \leq d$.
  The first version works in $O(n\log\sigma')$ time with $O(d)$ space where $\sigma' \leq d$ is the maximum number of distinct characters in the windows,
  and the second one works in $O(n + d\sigma)$ time with $(d+2)\sigma + O(d)$ space.
  We also show how our algorithms can be applied to efficient computation of minimal unique palindromic substrings (MUPS)
  and minimal absent palindromic words (MAPW) for a sliding window.
\end{abstract}
\section{Introduction}
\paragraph*{\bf Palindromes.}
A \emph{palindrome} is a string that reads the same forward and backward.
Palindromic structures in strings have been heavily studied
in the fields of string processing algorithms and combinatorics on strings~\cite{manacher1975new,groult2010counting,KosolobovRS13,rubinchik2018eertree,fici2014subquadratic,bannai2018diverse}.
One of the most famous results related to palindromic structures is
Manacher's algorithm~\cite{manacher1975new},
which computes all maximal palindromes in a given string $S$.
Manacher's algorithm essentially computes
all palindromes in $S$, since any palindromic substring of $S$
is a substring of some maximal palindrome in $S$.
Another interesting topic is
enumeration of distinct palindromes in a string.
It is known that any string of length $n$ contains
at most $n+1$ distinct palindromes including the empty string~\cite{droubay2001episturmian}.
Groult et al.~\cite{groult2010counting} proposed an $O(n)$-time algorithm
which enumerates all distinct palindromes in a given string
of length $n$ over an integer alphabet of size $\sigma = n^{O(1)}$.
For the same problem in the \emph{online} model,
Kosolobov et al.~\cite{KosolobovRS13} proposed
an $O(n\log\sigma)$-time and $O(n)$-space algorithm for a general ordered alphabet.
Kosolobov et al.'s algorithm is a combination of Manacher's algorithm and
Ukkonen's online suffix tree construction algorithm~\cite{ukkonen1995line}.
Rubinchik and Shur~\cite{rubinchik2018eertree} proposed
a new data structure called \emph{eertree},
which permits efficient access to distinct palindromes in a string
\emph{without} storing the string itself.
Eertrees can be utilized for solving problems related to palindromic structures,
such as the palindrome counting problem and
the palindromic factorization problem~\cite{rubinchik2018eertree}.
The size of the eertree of $S$ is linear in the number
$p_S$ of distinct palindromes in $S$~\cite{rubinchik2018eertree}.
It is known that $p_S$ is at most $|S|+1$,
and that it can be much smaller than the length $|S|$ of the string, e.g., for $S = (\mathtt{abc})^{n/3}$,
$p_S = 4$ since all distinct palindromes in $S$ are $\mathtt{a}$, $\mathtt{b}$, $\mathtt{c}$, and the empty string.
Thus, the size of the eertree of $S$ can be much smaller than
that of the suffix tree of $S$ which is $\Theta(n)$.
Therefore, it is of significance if one can build
eertrees without suffix trees.
Rubinchik and Shur~\cite{rubinchik2018eertree} indeed
proposed an online eertree construction algorithm
running in $O(n\log\sigma)$ time without suffix trees.

Recently, a concept of palindromic structures
called \emph{minimal unique palindromic substrings}~(\emph{MUPS}) is introduced
by Inoue et al.~\cite{inoue2018algorithms}.
A palindromic substring $w = S[i.. j]$ of a string $S$ is called a MUPS of $S$
if $w$ occurs in $S$ exactly once, and
$S[i+1.. j-1]$ occurs at least twice in $S$.
MUPSs are utilized for solving
the \emph{shortest unique palindromic substring} (\emph{SUPS})
problem~\cite{inoue2018algorithms},
which is motivated by an application in molecular biology.
Watanabe et al.~\cite{WatanabeNIBT20} proposed
an algorithm to solve the SUPS problem
based on the \emph{run-length encoding}~(\emph{RLE}) version of eertrees,
named $\mathtt{e^2rtre^2}$.

\paragraph*{\bf Sliding Window Model.}
In this paper, we consider the problems of computing palindromic structures for the \emph{sliding window} model.
The sliding window model is a natural generalization of the online model,
and the assumptions of this model are natural when we need to process a massive or a streaming string data with a limited memory space.
A typical and classical application to the sliding window model is data compression,
such as Lempel-Ziv 77~(the original version)~\cite{ziv1977universal} and PPM~\cite{cleary1984data}.
Note that sliding-window Lempel-Ziv 77 is an immediate application of suffix trees for a sliding window,
which can be maintained in $O(n\log\sigma')$ time using $O(d)$ space~\cite{fiala1989data,larsson1996extended,senft2005suffix}
where $d$ is the size of the window and $\sigma' \le d$ is the maximum number of distinct characters in every window.
Recently, several algorithms for computing substrings for a sliding window 
with certain interesting properties are proposed:
For instance, Crochemore et al.~\cite{crochemore2020absent} introduced
the problem of computing \emph{minimal absent words}~(\emph{MAWs}) for a sliding window, and proposed
an $O(n\sigma)$-time and $O(d\sigma)$-space algorithm
using suffix trees for a sliding window.
Mieno et al.~\cite{MienoKAFNIBT20} proposed an algorithm for computing \emph{minimal unique substrings}~(\emph{MUSs})~\cite{ilie2011minimum}
for a sliding window, in $O(n\log\sigma')$-time and $O(d)$ space,
again based on suffix trees for a sliding window.

\paragraph*{\bf Our Contributions.}
In this paper, we consider the problem of maintaining eertrees for the sliding window model, that is,
given a string $S$ of length $n$ and a \emph{window} of a fixed size $d$, we maintain eertrees of substrings $S[i.. i+d-1]$ for incremental $i = 0, 1, \ldots, n-d$.
Also, we consider the problem of maintaining MUPSs for a sliding window.
In addition, we introduce a new concept of palindromic structures called \emph{minimal absent palindromic words}~(\emph{MAPW}),
and consider the problem of maintaining MAPWs for a sliding window.
A string $w$ is called a MAPW of string $S$ if $w$ is a palindrome, $w$ does not occur in $S$, and $w[1..|w|-2]$ occurs in $S$.
MAPWs can be seen as a palindromic version of the notion of MAWs, which are extensively studied in the fields of string processing and
bioinformatics~\cite{crochemore2000data,mignosi2002words,chairungsee2012using,ota2014universal,fujishige2016computing}.

In this paper, we propose an algorithm which maintains eertrees for a sliding window in a total of $O(n\log\sigma')$ time using $O(d)$ space.
We then give an alternative eertree construction algorithm for a sliding window which runs in $O(n + d\sigma)$ time with $(d+2)\sigma + O(d)$ space.
As applications to the aforementioned result, we propose an algorithm which maintains MUPSs for a sliding window in a total of $O(n\log\sigma')$ time using $O(d)$ space, and
an algorithm which maintains MAPWs for a sliding window in a total of $O(n + d\sigma)$ time using $O(d\sigma)$ space.
We emphasize that our algorithms are stand-alone in the sense that they do not use suffix trees,
while the majority of existing efficient sliding window algorithms~(see above) make heavy use of suffix trees.

\section{Preliminaries} \label{sec:pre}
\subsection{Strings}
Let $\Sigma$ be an \emph{alphabet} of size $\sigma$.
An element of $\Sigma$ is called a \emph{character}.
An element of $\Sigma^*$ is called a \emph{string}.
The length of a string $S$ is denoted by $|S|$.
The \emph{empty string} $\varepsilon$ is the string of length 0.
If $S = xyz$, then $x$, $y$, and $z$ are called
a \emph{prefix}, \emph{substring}, and \emph{suffix} of $S$, respectively.
They are called a \emph{proper} prefix, \emph{proper} substring,
and \emph{proper} suffix of $S$
if $x \ne S$, $y \ne S$, and $z \ne S$, respectively.
If a non-empty string $x$ is both a proper prefix and a proper suffix of $S$,
then $x$ is called a \emph{border} of $S$.
For any  $0 \le i \le |S|-1$,
$S[i]$ denotes the $i$-th character of $S$.
For any  $0 \le i \le j \le |S|-1$,
$S[i.. j]$ denotes the substring of $S$
starting at position $i$ and ending at position $j$,
i.e., $S[i.. j] = S[i]S[i+1]\cdots S[j]$.
For convenience, $S[i.. j] = \varepsilon$ for any $i > j$.
A string $S$ is called a \emph{palindrome}
if $S[i] = S[|S|-i-1]$ for every $0 \le i \le |S|-1$.
Note that the empty string is a palindrome.
A substring $S[i.. j]$ of $S$ is said to be a \emph{palindromic substring} of $S$
if $S[i.. j]$ is a palindrome.
The \emph{center} of a palindromic substring $S[i.. j]$ of $S$ is $\frac{i+j}{2}$.
A palindromic substring $S[i.. j]$ of $S$ is \emph{maximal}
if $i = 0$, $j = |S|-1$, or $S[i-1..j+1]$ is not a palindrome.
We denote by $\lpp(S)$~(resp. $\lps(S)$) the longest palindromic prefix~(resp. suffix) of $S$.
We denote by $\DPal(S)$ the set of all distinct palindromes in $S$.
It is known that $|\DPal(S)| \le |S|+1$~\cite{droubay2001episturmian}.
For any non-empty strings $S$ and $w$, $w$ is said to be \emph{unique} in $S$ if $w$ occurs in $S$ exactly once.
Also, $w$ is said to be \emph{repeating} in $S$ if $w$ occurs in $S$ at least twice.
For convenience, we define that the empty string $\varepsilon$ is repeating in any string.
In what follows, we consider an arbitrary fixed string $S$ of length $n > 0$.
\subsection{Eertrees (Palindromic Trees)}
The \emph{eertree} of $S$ denoted by $\eertree(S)$ is a tree-like data structure
that enables us to efficiently access each of the distinct palindromes in $S$~\cite{rubinchik2018eertree}.
The $\eertree(S)$ consists of $m$ \emph{ordinary nodes}
and two auxiliary nodes, denoted \emph{$\zeronode$} and \emph{$\minusonenode$},
where $m = |\DPal(S)| -1 $.
Each ordinary node corresponds to
each element of $\DPal(S)\setminus\{\varepsilon\}$.
For each ordinary node $v$,
we denote by $\pal(v)$ the palindrome which corresponds to $v$,
and by $\len(v)$ its length.
Conversely, for each non-empty palindromic substring $p$ of $S$,
we denote by $\node(p)$ the node which corresponds to the palindrome $p$.
Namely, $\node(\pal(v)) = v$ for each ordinary node $v$.
For convenience,
we define $\pal(\zeronode) = \pal(\minusonenode) = \varepsilon$,
$\len(\zeronode) = 0$, and $\len(\minusonenode) = -1$.
For any nodes $u, v$ in $\eertree(S)$,
there is an edge $(u, v)$ if and only if
$\len(u) + 2 = \len(v)$ and
$\pal(u) = \pal(v)[1..\len(v)-2]$.
Each edge $(u, v)$ is labeled by a character $\pal(v)[0]$.
Also, each node $v$ in $\eertree(S)$ has a \emph{suffix link}
denoted by $\sufLink(v)$.
For each node $v$ in $\eertree(S)$ with $\len(v) \ge 2$,
$\sufLink(v)$ points to the node corresponding to the longest palindromic proper suffix of $\pal(v)$.
For each node $v$ in $\eertree(S)$ with $\len(v) = 1$,
$\sufLink(v)$ points to the $\zeronode$.
Also, $\sufLink(\zeronode) = \minusonenode$ and
$\sufLink(\minusonenode) = \minusonenode$.
For each node $v$ in $\eertree(S)$,
$\inSufLink(v) = |\{u\mid \sufLink(u) = v\}|$
denotes the number of incoming suffix links of $v$.
See Fig.~\ref{fig:eertree} for an example of $\eertree(S)$.
\begin{figure}[t]
  \centering
  \includegraphics[width=0.7\linewidth]{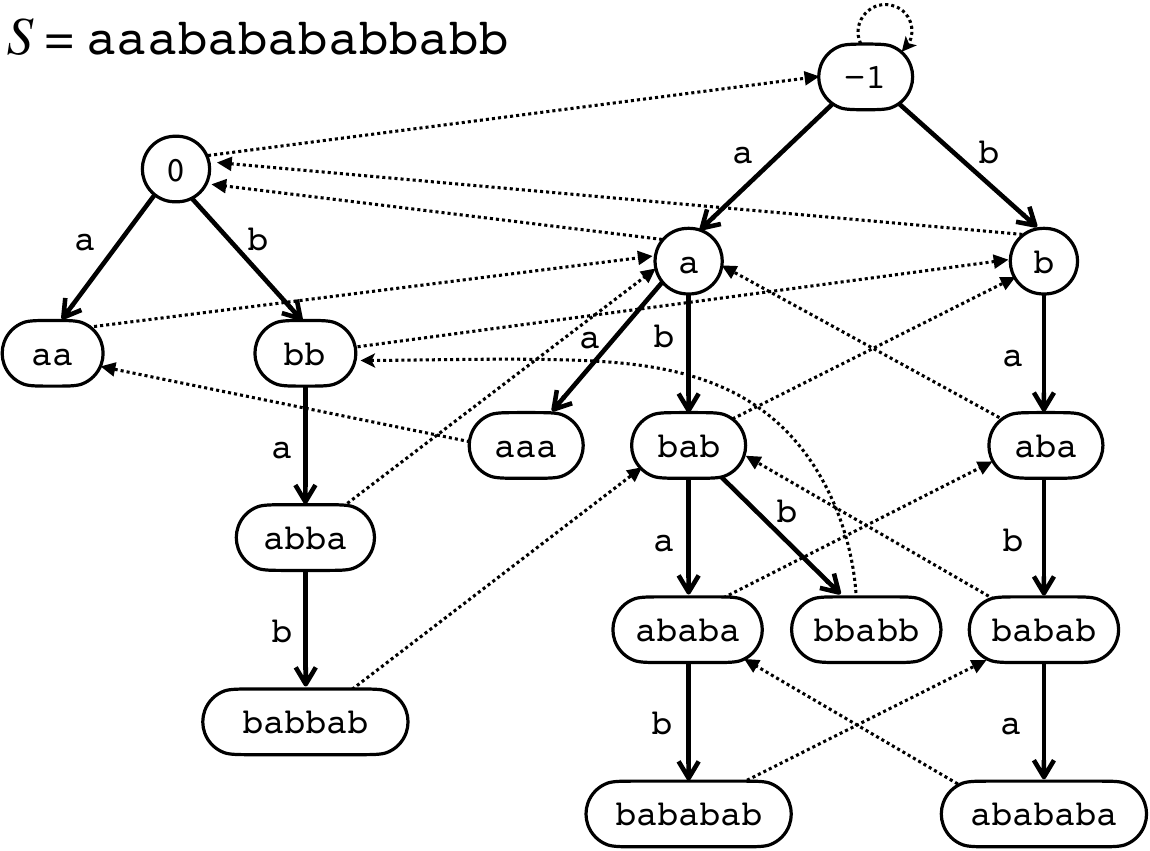}
  \caption{
    The eertree of $S = \mathtt{aaababababbabb}$.
    The solid and broken arrows represent edges and suffix links, respectively.
    Note that $\pal(v)$ is written inside each node $v$ in this figure,
    however, it is for only explanation.
    Namely, each node does not explicitly store the corresponding string.
  } \label{fig:eertree}
\end{figure}

Note that
each node $v$ does not store the string $\pal(v)$ explicitly.
Instead, we can obtain $\pal(v)$
by traversing edges backward, from $v$ to the root,
since $\pal(u) = c\pal(u')c$
for each node $u$ with $|\pal(u)| \ge 2$
where $u'$ is the parent of $u$
and $c$ is the label of the edge $(u',u)$.
Each node only stores pointers to its children
and a constant number of integers.
Thus, the size of $\eertree(S)$ is linear in
the number of nodes, i.e., $O(|\DPal(S)|)$.
It is known that $\eertree(S)$ can be constructed in
$O(n\log\sigma)$ time for any string $S$
given in an online manner~\cite{rubinchik2018eertree}.

\subsection{Sliding Window} \label{subsec:slidingwindow}
We formalize sliding windows over string $S$.
For each time $t = 0, 1, \ldots$, we consider the substring $S[i_t.. j_t]$ called \emph{the window at time $t$}.
The windows must satisfy the following conditions:
(1) $i_0 = j_0 = 0$ for the initial window at time $0$; and
(2) $0 \le i_t \le j_t \le n-1$ and either $(i_t, j_t) = (i_{t-1}+1, j_{t-1})$
or $(i_t, j_t) = (i_{t-1}, j_{t-1}+1)$ for every time $t > 0$.
In other words, the second condition means that we can either
\emph{delete} the leftmost character from the current window, or
\emph{append} a character to the right end of the current window at each time.

Given a sequence of windows~(or equivalently, a sequence of delete / append operations),
the aim of our sliding window model is processing the windows in space proportional to the size of each window.
This paper mainly deals with the problem of maintaining eertrees with respect to a sequence of windows over a given string $S$.

\section{Combinatorial Properties on Palindromes for a Sliding Window} \label{sec:comb}

In this section, we show some combinatorial properties
on palindromes for a sliding window, which is helpful for designing
efficient algorithms to maintain eertrees for a sliding window.
Since the nodes of the eertree of a string represent
all distinct palindromes in the string, we obtain the next lemma.
\begin{lemma} \label{lem:removed_node}
  There is a node $\ell$ in $\eertree(S[i-1.. j-1])$ to be removed
  when the leftmost character $S[i-1]$ is deleted from $S[i-1.. j-1]$
  if and only if
  (A) $\pal(\ell)$ is unique in $S[i-1.. j-1]$,
  (B) $\pal(\ell) = \lpp(S[i-1.. j-1])$, and
  (C) $\ell$ is a leaf node.
\end{lemma}
\begin{proof}
  ($\Rightarrow$)
  (A) Since $\ell$ is removed, $\pal(\ell)$ does not occur in $S[i.. j-1]$.
  Thus, $\pal(\ell)$ occurs in $S[i-1.. j-1]$ only as a prefix, i.e.,
  $\pal(\ell)$ is unique in $S[i-1.. j-1]$.
  (B) Assume that $\pal(\ell)$ is shorter than $\lpp(S[i-1.. j-1])$.
  Then, $\pal(\ell)$ is a proper prefix of $\lpp(S[i-1.. j-1])$.
  Also, $\pal(\ell)$ is a proper suffix of $\lpp(S[i-1.. j-1])$
  since $\lpp(S[i-1.. j-1])$ is a palindrome.
  This contradicts that $\pal(\ell)$ is unique in $S[i-1.. j-1]$.
  Thus, $\pal(\ell) = \lpp(S[i-1.. j-1])$.
  (C) If we assume that $\ell$ has a child, then
  $\pal(\ell)$ has an occurrence in $S[i-1..j-1]$
  that is not a prefix of $S[i-1..j-1]$, a contradiction.
  \noindent($\Leftarrow$)
  Since $\pal(\ell)$ is a palindromic prefix of $S[i-1.. j-1]$
  and unique in $S[i-1.. j-1]$, $\pal(\ell)$ does not occur in $S[i.. j-1]$.
  Thus, $\ell$ is removed when $S[i-1]$ is deleted.
\end{proof}

Namely, when the leftmost character of the window is deleted, at most one leaf will be removed from the eertree.
Also, in order to detect such a leaf, we need to compute the longest palindromic prefix of each window and to determine its uniqueness.
In the following, we show some combinatorial properties on unique palindromes and the longest palindromic prefix for a sliding window.

\subsection{Unique Palindromes for a Sliding Window}
A palindromic substring $w$ of string $S$ is said to be \emph{left-maximal} in $S$
if there is no palindromic substring of $S$ which contains $w$ as a proper suffix.
See Fig.~\ref{fig:left_maximal} for examples.
If a palindrome $w$ is not left-maximal in $S$, then for some palindrome $w'$, $w$ is a proper suffix and prefix of $w'$, i.e., $w$ is not unique in $S$.
In other words, any unique palindromic substring must be left-maximal.
\begin{figure}[h]
  \centering
  \includegraphics[width=0.6\linewidth]{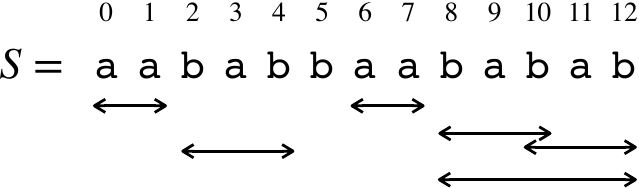}
  \caption{
    For string $S = \mathtt{aababbaababab}$, its palindromic substring $\mathtt{aa}$ is left-maximal in $S$.
    On the other hand, $\mathtt{bab}$ is not left-maximal in $S$
    since there is a palindromic substring $S[8.. 12] = \mathtt{babab}$ of $S$ which contains $\mathtt{bab}$ as a proper suffix.
  } \label{fig:left_maximal}
\end{figure}
\begin{lemma}\label{lem:leftmaximal}
  For any time $t$ and any left-maximal palindromic substring $w$ of $S[i_t.. j_t]$,
  there exists time $t' < t$ which satisfies one of the followings:
  \begin{enumerate}
    \item the longest palindromic suffix of $S[i_{t'}.. j_{t'}]$ is $w$, or
    \item the longest palindromic suffix of $\lpp(S[i_{t'}.. j_{t'}])$ is $w$.
  \end{enumerate}
\end{lemma}
\begin{proof}
  Let $w = S[s..e]$.
  For the sake of contradiction, we assume the contrary.
  Namely, for every time $t' < t$,
  neither the longest palindromic suffix of $S[i_{t'}.. j_{t'}]$ nor
  the longest palindromic suffix of $\lpp(S[i_{t'}.. j_{t'}])$ is equal to $S[s..e]$.
  Consider a window in the past such that its ending position is $e$.
  Since the longest palindromic suffix of the window is not $S[s..e]$,
  there is another palindromic suffix ending at $e$ which is longer than $S[s.. e]$.
  Now let $v = S[s'..e]$ be the shortest one among such palindromic suffixes.
  Then, the longest palindromic suffix of $v$ is $w$.
  Next, consider a window $S[i_{\tilde{t}}.. j_{\tilde{t}}]$ at time $\tilde{t} < t$ with its starting position $i_{\tilde{t}} = s'$.
  If $v$ is the longest palindromic prefix of the window $S[i_{\tilde{t}}.. j_{\tilde{t}}]$,
  then $w$ becomes the longest palindromic suffix of $v = \lpp(S[i_{\tilde{t}}.. j_{\tilde{t}}])$, however, it contradicts our assumption.
  Thus, there is another palindromic prefix starting at $s'$ which is longer than $v$.
  Now let $u = S[s'..e']$ be the longest one among such palindromic prefixes.
  Since $v$ is a palindromic proper prefix of $u$, it is also a palindromic proper suffix of $u$.
  Let $s''$ be the starting position of $v$ which occurs as a suffix of $u$.
  We further consider two sub-cases~(see also Fig.~\ref{fig:left_maximal_contradiction}):
  \begin{itemize}
    \item[(a)] If $s'' < s$, then the border $S[s''.. e]$ of $v$ is a palindrome.
      Also, $S[s''.. e]$ is longer than $S[s..e]$ and is shorter than $S[s'..e]$.
      This contradicts that the minimality of $v = S[s'.. e]$.
    \item[(b)] If $s'' \ge s$, then $i_t \le s \le s''$ and $e' \le j_{\tilde{t}} \le j_t$,
      and thus, $v = S[s''..e']$ is a substring of the current window $S[i_t.. j_t]$.
      This contradicts that $w$ is left-maximal in $S[i_t.. j_t]$.
  \end{itemize} 
  Therefore, we have proved the lemma.
\end{proof}
\begin{figure}[h]
  \centering
  \includegraphics[width=\linewidth]{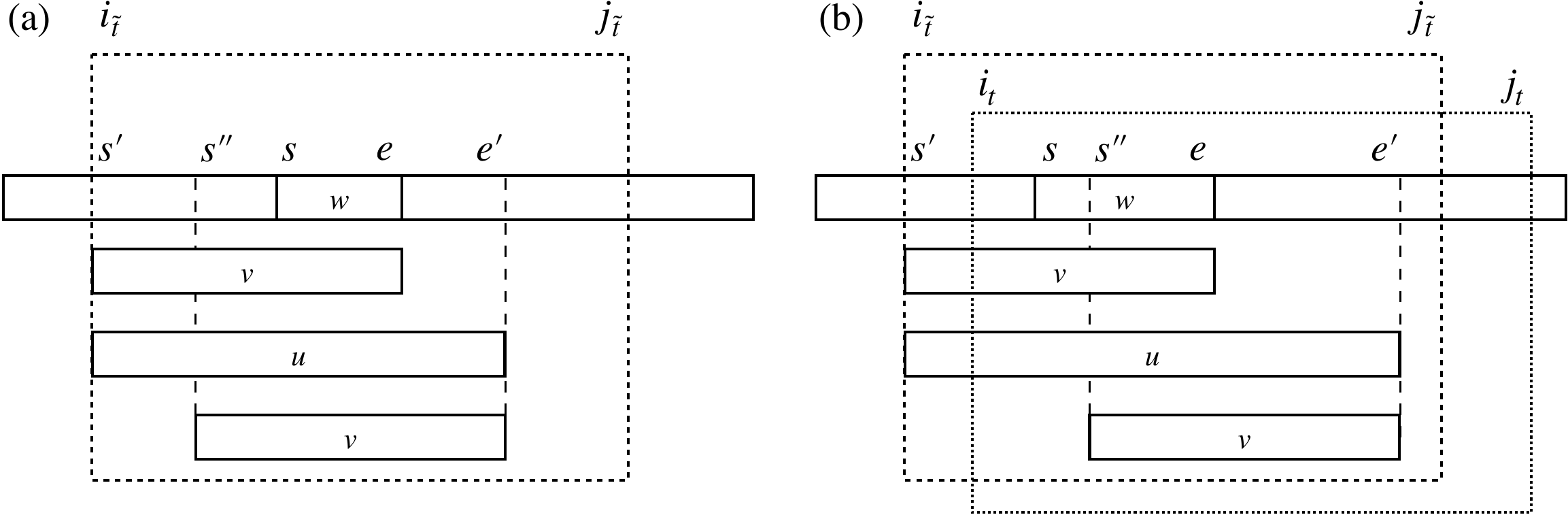}
  \caption{
    Illustration for contradictions in the proof of Lemma~\ref{lem:leftmaximal}.
  } \label{fig:left_maximal_contradiction}
\end{figure}
As mentioned before, any unique palindrome is left-maximal.
Thus, Lemma~\ref{lem:leftmaximal} is useful for maintaining uniqueness of palindromes for a sliding window.

\subsection{Longest Prefix Palindrome for a Sliding Window}
Next, we consider the longest palindromic prefixes for sliding windows.
\begin{lemma}\label{lem:longestprefixpal}
  Let $w$ be the longest palindromic prefix of the window $S[i_t.. j_t]$ at time $t$.
  There exists time $t' < t$ which satisfies one of the followings:
  \begin{enumerate}
    \item the longest palindromic suffix of $S[i_{t'}.. j_{t'}]$ is $w$, or
    \item the longest palindromic suffix of $\lpp(S[i_{t'}.. j_{t'}])$ is $w$.
  \end{enumerate}
\end{lemma}
\begin{proof}
  Let $w = S[s.. e]$. For the sake of contradiction, we assume the contrary.
  Namely, for every time $t' < t$,
  neither the longest palindromic suffix of $S[i_{t'}.. j_{t'}]$ nor
  the longest palindromic suffix of $\lpp(S[i_{t'}.. j_{t'}])$ is not equal to $S[s.. e]$.
  Similar to the proof of Lemma~\ref{lem:leftmaximal},
  let $v = S[s'..e]$ be the shortest palindrome which is ending at $e$ and is longer than $w = S[s.. e]$.
  Namely, $w$ is the longest palindromic suffix of $v$.

  Next, consider a window $S[i_{\tilde{t}}.. j_{\tilde{t}}]$ at time $\tilde{t} < t$ with its starting position $i_{\tilde{t}} = s'$.
  Similar to the proof of Lemma~\ref{lem:leftmaximal}, again,
  let $u = S[s'..e']$ be the longest palindromic substring of the window $S[i_{\tilde{t}}.. j_{\tilde{t}}]$.
  Further let $c_w$, $c_v$, and $c_u$ be respectively the center of $w$, $v$, and $u$.
  From the assumptions, $c_v < c_w$ and $c_v < c_u$ hold.
  Next, we consider three sub-cases~(see also Fig.~\ref{fig:longest_prefix_pal_contradiction}):
  \begin{itemize}
    \item[(a)] If $c_u < c_w$, then the palindrome $u'$ ending at $e$ whose center equals $c_u$, is longer than $w$ and is shorter than $v$.
      This contradicts that $w$ is the longest palindromic suffix of $v$.
    \item[(b)] If $c_u = c_w$, then $|v| = e-s'+1 = e'-s+1$.
      Thus, $S[s.. e'] = S[s'..e] = v$ since $S[s'.. e']$ is a palindrome.
      This contradicts that $w$ is the longest palindromic prefix of the current window $S[i_t.. j_t] = S[s.. j_t]$.
    \item[(c)] If $c_w < c_u$, then the palindrome $u''$ starting at $s$ whose center equals $c_u$, is longer than $w$.
      This again contradicts that $w$ is the longest palindromic prefix of the current window $S[i_t.. j_t] = S[s.. j_t]$.
  \end{itemize} 
  Therefore, we have proved the lemma.
\end{proof}
\begin{figure}[h]
  \centering
  \includegraphics[width=\linewidth]{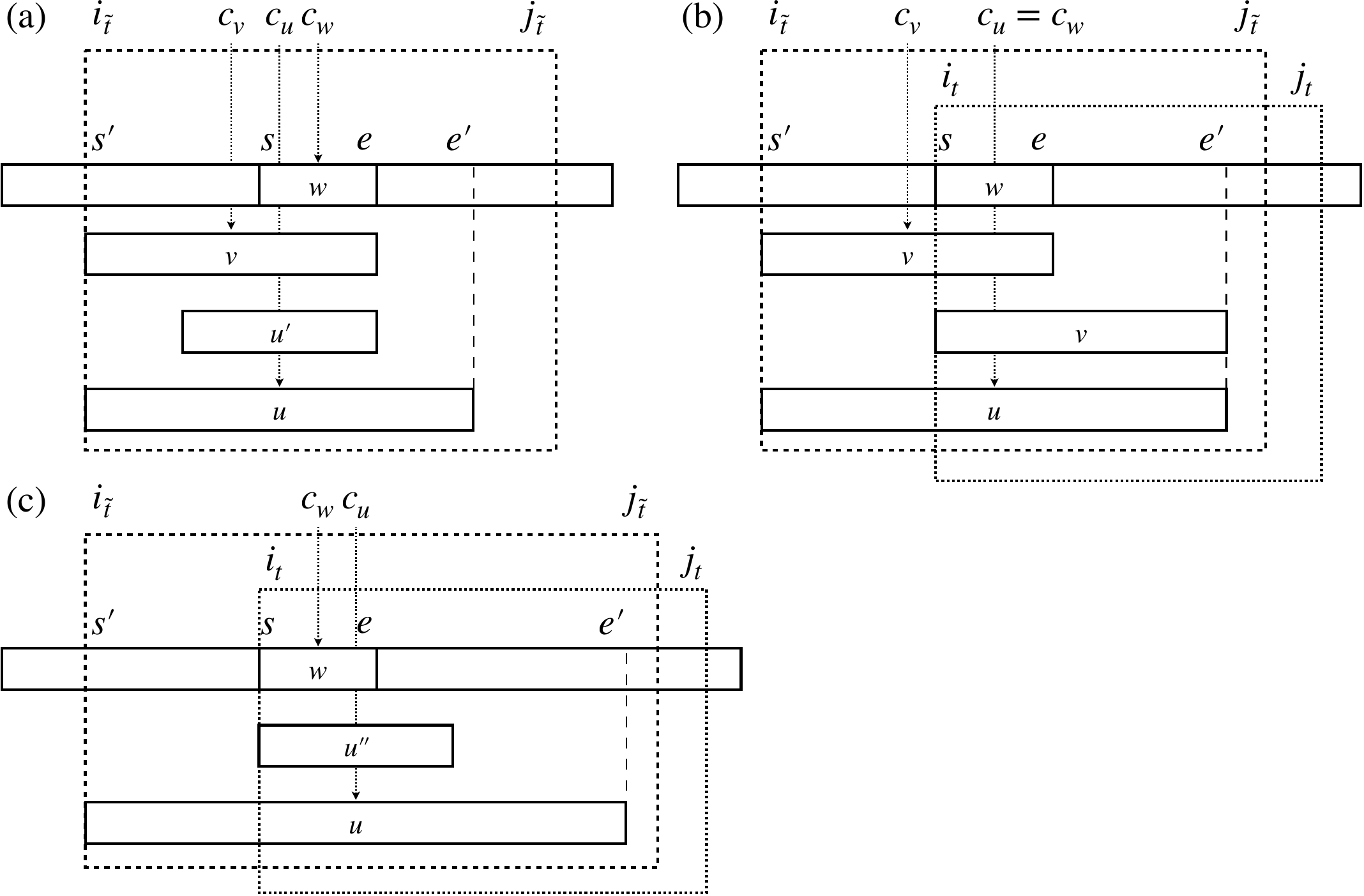}
  \caption{
    Illustration for contradictions in the proof of Lemma~\ref{lem:longestprefixpal}.
  }\label{fig:longest_prefix_pal_contradiction}
\end{figure}

\section{Eertree for a Sliding Window} \label{sec:eertree}
In this section, we show how to update a given eertree when we shift the sliding window to the right by one character.
Sliding a given window consists of two operations:
\emph{deleting} the leftmost character and \emph{appending} a character to the right end.
Namely, when the eertree of $S[i-1.. j-1]$ is given,
we first compute the eertree of $S[i.. j-1]$~(deleting the leftmost character $S[i-1]$),
and then, compute the eertree of $S[i.. j]$~(appending a character $S[j]$).
To update the eertree when a character is appended,
we can apply Rubinchik and Shur's 
algorithm~\cite{rubinchik2018eertree}
which constructs the eertree of a given string in an online manner.
In this section, we propose new additional data structures and algorithms
which update the eertree when the leftmost character is deleted.

We emphasize that our algorithms work for any valid\footnote{C.f., the definition of sliding windows in Section~\ref{subsec:slidingwindow}.} sequence of windows of arbitrary lengths.
However, for simplicity, we consider the case where a fix-sized window of length $d$ shifts to the right one-by-one throughout this section.
\subsection{Auxiliary Data Structures for Detecting the Node to be Deleted}
We introduce auxiliary data structures for computing the longest palindromic prefixes and for determining uniqueness of palindromes.
\paragraph{For Computing the Longest Palindromic Prefix.}
Let $\prefPal[0.. d-1]$ be a cyclic array of size $d$
such that $\prefPal[i_t \bmod d]$ stores the node which corresponds to the longest palindromic prefix of the window $S[i_t.. j_t]$ at each time $t$.
Namely, for every time $t$, $\prefPal[i_t\bmod d] = \node(\lpp(S[i_t.. j_t]))$ holds.
\paragraph{For Determining Uniqueness of a Palindrome.}
For each ordinary node $v$ in $\eertree(S[i.. j])$,
let $\rightmost_{i,j}(v)$ be the starting position of the rightmost occurrence of $\pal(v)$ in $S[i..j]$.
Further let $\secondrightmost_{i,j}(v)$ be the starting position of the second rightmost occurrence of $\pal(v)$ in $S[i..j]$ if such a position exists,
and otherwise, $\secondrightmost_{i,j}(v) = -1$.
Throughout the computation of the eertree for a sliding window,
for each node $v$ of $\eertree(S[i.. j])$ we keep the following invariant
$\BegPair_{i,j}(v)$ which consists of two fields $\first$ and $\second$ such that:
\begin{eqnarray*}
    \BegPair_{i,j}(v).\first & = & \begin{cases}
      \rightmost_{i,j}(v) & \text{if } \inSufLink(v) = 0,\\
    -1\text{ or some occurrence}\\
    \qquad\text{of }\pal(v)\text{ in }S[0.. j]& \text{otherwise.}
    \end{cases} \\
    \BegPair_{i,j}(v).\second & = & \begin{cases}
      \secondrightmost_{i,j}(v) & \text{if } \inSufLink(v) = 0\\
                                &\text{ and }\secondrightmost_{i.j}(v) \ne -1,\\
      -1\text{ or some occurrence}\\
      \qquad\text{of }\pal(v)\text{ in }S[0.. j]& \text{otherwise.}
    \end{cases}
\end{eqnarray*}
Namely, $\BegPair_{i,j}(v)$ stores the rightmost and second rightmost occurrences of $\pal(v)$ in $S[i.. j]$ when $\inSufLink(v) = 0$, if such occurrences exist.
Otherwise, it temporarily stores some pair of integers, however, it will never be referred in our algorithms.
In other words, we employ a kind of lazy maintenance of the rightmost and second rightmost occurrences of $\pal(v)$ in $S[i..j]$ that suffices for our purpose.
See Fig.~\ref{fig:BegW} for an example of $\BegPair_{i,j}(v)$.
\begin{figure}[h]
  \centering
  \includegraphics[width=0.5\linewidth]{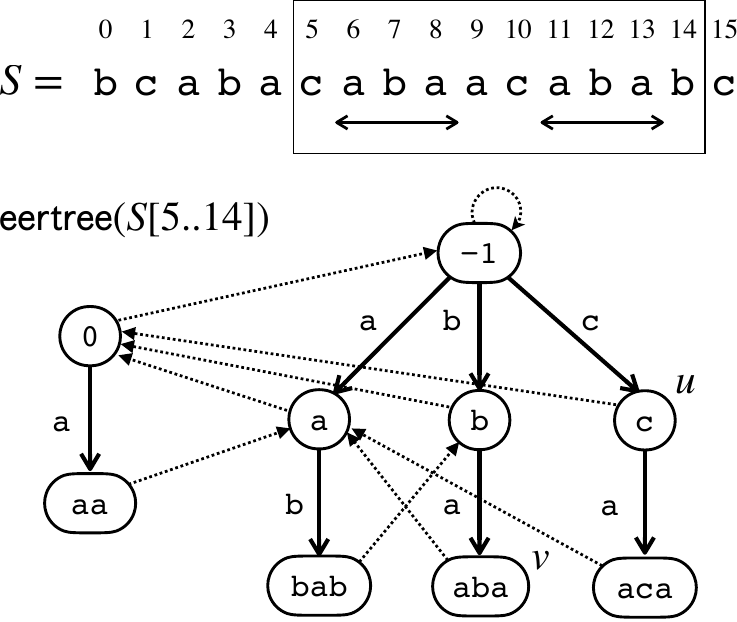}
  \caption{
    Examples for $\BegPair_{i,j}(v)$.
    For string $S = \mathtt{bcabacabaacababc}$ and window $[5, 14]$, the $\eertree(S[5.. 14])$ is depicted.
    Consider node $v$ in $\eertree(S[5.. 14])$ with $\pal(v) = \mathtt{aba}$.
    The rightmost and the second rightmost occurrences of $\mathtt{aba}$ in the window $S[5.. 14]$ are $11$ and $6$.
    Namely, $\rightmost_{5,14}(v) = 11$ and $\secondrightmost_{5, 14}(v) = 6$.
    Further, $\inSufLink(v) = 0$, and thus, $\BegPair_{5, 14}(v) = (11, 6)$.
    Also, for node $u$ in $\eertree(S[5.. 14])$ with $\pal(u) = \mathtt{c}$,
    $\BegPair_{5, 14}(u) = (10, 5)$ since $\rightmost_{5,14}(u) = 10$, $\secondrightmost_{5, 14}(u) = 5$, and $\inSufLink(u) = 0$.
    When the leftmost character $S[5] = \mathtt{c}$ is deleted from the window $S[5.. 14]$, $\secondrightmost_{5,14}(u)$ changes to $-1$.
    However, $\BegPair_{6, 14}(u) = \BegPair_{5, 14}(u) = (10, 5)$ is allowed since $\BegPair_{6, 14}(u).\second = 5 < 6$ is a valid value for our invariant.
    Namely, we do not have to update $\BegPair_{6, 14}(u)$ explicitly
    when deleting the leftmost character $S[5]$ from $S[5.. 14]$. 
  } \label{fig:BegW}
\end{figure}

The next lemma states that given a node $v$, we can determine if $\pal(v)$ is unique or not by checking the incoming suffix links of $v$ and $\BegPair_{i,j}(v)$.
\begin{lemma}\label{lem:unique}
  Let $v$ be any node in $\eertree(S[i.. j])$.
  Then, $\pal(v)$ is unique in $S[i.. j]$ if and only if
  $\inSufLink(v) = 0$ and $\BegPair_{i,j}(v).\second < i$.
\end{lemma}
\begin{proof}
  ($\Rightarrow$)
  We show the contraposition. There are two cases:
  (1) If $\inSufLink(v) \ne 0$, then there is a palindromic substring $P$ of $S[i..j]$ with $\lps(P) = \pal(v)$.
  Namely, $\pal(v)$ is not left-maximal in $S[i..j]$, and thus, $\pal(v)$ is not unique in $S[i.. j]$.
  (2) If $\inSufLink(v) = 0$ and $\BegPair_{i,j}(v).\second = \secondrightmost_{i,j}(v) \ge i$,
  then $\pal(v)$ occurs at least twice in $S[i.. j]$ at positions $\BegPair_{i,j}(v).\second$ and $\BegPair_{i,j}(v).\first$.\\
  \noindent($\Leftarrow$)
  For the sake of contradiction, we assume that $\pal(v)$ is not unique in $S[i..j]$.
  Then, by the definition of $\secondrightmost$, $i \le \secondrightmost_{i,j}(v) \le j$.
  However, since $\inSufLink(v) = 0$ and $\secondrightmost_{i,j}(v) \ne -1$,
  $\secondrightmost_{i,j}(v) = \BegPair_{i,j}(v).\second < i$, a contribution.
\end{proof}
Next, we introduce our algorithms to maintain $\prefPal$ and $\BegPair$ for a sliding window
which utilizes combinatorial properties shown in Section~\ref{sec:comb}.
\subsection{Maintaining the Auxiliary Data Structures}
First, in Algorithm~\ref{alg:1}, we show subroutine $\updatebp$ which updates the member variable $v.\bp$ of a given node $v$
where $v.\bp$ must be kept equal to $\BegPair_{i_t,j_t}(v)$ at each time $t$.
It will be called in the algorithms that we show later.

\begin{algorithm}[t]
  \caption{$\updatebp(v, x)$.}
  \begin{algorithmic}[1]\label{alg:1}
    \REQUIRE Node $v$, and a starting position $x$ of $\pal(v)$.
    \ENSURE Update $v.\bp$ appropriately with respect to the position $x$.
    \IF {$x > v.\bp.\first$}
      \STATE $v.\bp.\second \leftarrow v.\bp.\first$
      \STATE $v.\bp.\first \leftarrow x$
    \ELSIF {$x > v.\bp.\second$}
      \STATE $v.\bp.\second \leftarrow x$
    \ENDIF
  \end{algorithmic}
\end{algorithm}
Next, we show our algorithms for updating data structures when we slide the given window.
When the leftmost character $S[i-1]$ is deleted from $S[i-1..j-1]$, our data structures are updated by Algorithm~\ref{alg:2}.
Also, when a character $S[j]$ is appended to $S[i..j-1]$, our data structures are updated by Algorithm~\ref{alg:3}.
\begin{algorithm}[t]
  \caption{Update $\BegPair$ and $\prefPal$ when the leftmost character is deleted.}
  \begin{algorithmic}[1]\label{alg:2}
    \REQUIRE
      $\lpsuf = \node(\lps(S[i-1.. j-1]))$, and\\
      $v.\bp = \BegPair_{i-1,j-1}(v)$ for each node $v$ in $\eertree(S[i-1.. j-1])$.
    \ENSURE
      $\lpsuf = \node(\lps(S[i.. j-1]))$, and\\
      $v.\bp = \BegPair_{i,j-1}(v)$ for each node $v$ in $\eertree(S[i.. j-1])$.
    \STATE $\lppref \leftarrow \prefPal[i-1]$
      \qquad $\mathtt{\backslash\backslash }\pal(\lppref) = \lpp(S[i-1.. j-1])$
    \IF {$\lppref = \lpsuf$}
        \STATE $\lpsuf \leftarrow \sufLink(\lpsuf)$
        \qquad$\mathtt{\backslash\backslash~}$\texttt{For the case the window is a palindrome}\\
    \ENDIF
    \STATE $q \leftarrow\sufLink(\lppref)$
    \STATE $\inSufLink(q) \leftarrow \inSufLink(q) - 1$
    \STATE $x \leftarrow i - 1 + \len(\lppref) - \len(q)$
      \qquad$\mathtt{\backslash\backslash}~x~$\texttt{is a starting position of }$\pal(q)$
    \STATE $\updatebp(q, x)$
    \IF {$\len(q) > \len(\prefPal[x])$}
      \STATE $\prefPal[x] = q$
    \ENDIF
    \IF {$\inSufLink(\lppref) = 0$ and $\lppref.\bp.second < i-1$.}
      \STATE Remove node $\lppref$ from the eertree
    \ENDIF
  \end{algorithmic}
\end{algorithm}
\begin{algorithm}[t]
  \caption{Update $\BegPair$ and $\prefPal$ when a character is appended.}
  \begin{algorithmic}[1] \label{alg:3}
    \REQUIRE
      $\lpsuf = \node(\lps(S[i.. j-1]))$, $S[j]$, and\\
      $v.\bp = \BegPair_{i,j-1}(v)$ for each node $v$ in $\eertree(S[i.. j-1])$.
    \ENSURE
      $\lpsuf = \node(\lps(S[i.. j]))$, and\\
      $v.\bp = \BegPair_{i,j}(v)$ for each node $v$ in $\eertree(S[i.. j])$.
    \STATE Compute $\lps(S[i.. j])$ and overwrite $\lpsuf \leftarrow \node(\lps(S[i.. j]))$
    \IF {$\lpsuf$ does not exist in $\eertree(S[i.. j-1])$}
      \STATE Add new node $\lpsuf$ to the eertree
    \ENDIF
    \STATE $y \leftarrow j-\len(\lpsuf)+1$
      \qquad$\mathtt{\backslash\backslash}~y~$\texttt{is a starting position of }$\lps(S[i..j])$
    \STATE $\updatebp(\lpsuf, y)$
    \STATE $\prefPal[y] \leftarrow \lpsuf$
  \end{algorithmic}
\end{algorithm}

\paragraph{Time Complexities.}
Clearly, Algorithm~\ref{alg:1} runs in constant time.
In Algorithm~\ref{alg:2}, all lines except for Line~13 can be processed in constant time.
Thus, the total running time of Algorithm~\ref{alg:2} is dominated by Line~13, i.e., $O(\log\sigma')$.
In Algorithm~\ref{alg:3},
the first four lines can be processed in amortized $O(\log\sigma')$ time by using the online construction algorithm~\cite{rubinchik2018eertree}.
Also, the remaining lines can be processed in constant time, and thus,  the total running time of Algorithm~\ref{alg:3} is amortized $O(\log\sigma')$.

\paragraph{Correctness.}
First, it is clear that Algorithm~\ref{alg:1} runs correctly.

Next, let us consider the correctness of Algorithm~\ref{alg:2}.
Let us first consider a special case when the window $S[i-1.. j-1]$ itself is a palindrome.
Then, we need to update $\lpsuf$, which will be used in Algorithm~\ref{alg:3}.
Lines~2--3 of Algorithm~\ref{alg:2} captures such a case.
Next, we show that $\BegPair$ for all nodes are updated correctly.
By Lemma~\ref{lem:leftmaximal}, it is suffice to update $v.\bp$ for every node $v$ where $\pal(v)$ is left-maximal.
Let $q$ be is the node corresponding to the longest palindromic suffix of $\lppref = \lpp(S[i-1.. j-1])$.
Then, it is suffice to update $q.\bp$ since the node $q$ is the only candidate for a node whose corresponding palindrome to be left-maximal \emph{just} in this step.
Thus, we update only $q.\bp$ in Lines~5--8, if it is needed.
Further, we show that $\prefPal$ is also updated correctly.
By Lemma~\ref{lem:longestprefixpal},
the longest palindromic prefix of a window must be the longest palindromic suffix of either some window or the longest palindromic prefix of some window.
The palindrome $\pal(q)$ is the only one which is to be such a palindrome \emph{just} in this step.
Thus, $\prefPal[x]$ is the only candidate which may be updated in this step
where $x$ is the starting position of the occurrence of $\pal(q)$
which is the longest palindromic suffix of $\lpp(S[i-1, j-1])$.
Therefore, it is suffice to update $\prefPal[x]$ and update it if necessary (Lines~9--11).
Line~12 determines the uniqueness of $\lpp(S[i-1.. j-1])$ correctly by using Lemma~\ref{lem:unique},
and if it is unique, then the corresponding node $\lppref$ is removed (in Line~13).

Finally, consider the correctness of Algorithm~\ref{alg:3}.
When a character is appended, we first check the new longest palindromic suffix,
and create a new node corresponding to the palindrome if necessary.
These procedures in Lines~1--4 are correctly performed by running the online construction algorithm~\cite{rubinchik2018eertree}.
Let $y$ be the starting position of the longest palindromic suffix of the window $S[i..j]$.
The palindrome $\lps(S[i..j])$ is the only candidate for a palindrome to be left-maximal \emph{just} in this step,
and thus, by Lemma~\ref{lem:leftmaximal}, it is suffice to update $\lpsuf.\bp$ in this step.
Also, by Lemma~\ref{lem:longestprefixpal},
$\lpsuf$ is the only candidate for the node that we need to newly store into $\prefPal$ in this step.
At this moment, $\lpsuf$ is clearly the longest palindrome starting at position $y$.
Thus, we set $\prefPal[y] \leftarrow \lpsuf$~(Line~7).

To summarize this section, we obtain the following theorem.
\begin{theorem} \label{thm:main}
  We can maintain eertrees for a sliding window
  in a total of $O(n\log\sigma')$ time using $O(\dpal) + d$ space
  where $\dpal \le d$ be the maximum number of distinct palindromes in all windows.
\end{theorem}
\begin{proof}
  When a character is appended to the right end of the window,
  we update the eertree itself by applying the online algorithm~\cite{rubinchik2018eertree},
  and update our auxiliary data structures by using Algorithm~\ref{alg:3}.
  When the leftmost character is deleted from the window,
  we update the eertree and the auxiliary data structures by using Algorithm~\ref{alg:2}.
  The total running time is $O(n\log\sigma')$.
  The space usage is $O(\dpal)$ words for the original eertree and the auxiliary member variable $v.\bp$ for each node $v$ of the eertree,
  plus $d$ words for the array $\prefPal[0..d-1]$.
\end{proof}

By applying a subtle modification to the above algorithm,
we obtain another variant of the algorithm (Theorem~\ref{thm:eertree_nsigma} below) which is faster than Theorem~\ref{thm:main}
when $\dpal \sigma < n \log \sigma'$, but using additional $(\dpal+1)\sigma$ space.

\begin{theorem} \label{thm:eertree_nsigma}
  We can maintain eertrees for a sliding window
  in a total of $O(n + \dpal\sigma)$ time using $(\dpal+1)\sigma + O(\dpal) + d \in O(d\sigma)$ space.
\end{theorem}
\begin{proof}
  In the original eertrees,
  each node stores a binary search tree to maintain branches dynamically.
  Instead, we use an array of integers of size $\sigma$,
  which allows to add, delete, and search for a node pointer~(i.e., edge)
  labeled by a given character in constant time.
  Thus, the $\log\sigma'$ factor in our time complexity can be removed.
  On the other hand, we need $\sigma + O(1)$ space
  to represent each node object,
  and $\Theta(\sigma)$ time to initialize it.
  If we naively initialize such a node object when adding a new node,
  the total time complexity increases to $O(n\sigma)$.
  However, we can \emph{reuse} node objects
  that had been removed when deleting a character
  since such removed nodes and new nodes to be added are leaves,
  i.e., they do not have any child~(Lemma~\ref{lem:removed_node}).
  Thus, by reusing node objects,
  we do not need to initialize an array of size $\sigma$
  when adding a new leaf node.
  The total number of node objects to initialize is $\dpal+1$,
  and it costs $O(\dpal\sigma)$ total time to initialize them.
\end{proof}

\section{Applications of Eertrees for a Sliding Window}
In this section, we apply our sliding-window eertree
algorithm of Section~\ref{sec:eertree}
to computing minimal unique palindromic substrings and
minimal absent palindromic words for a sliding window.

\subsection{Computing Minimal Unique Palindromic Substrings for a Sliding Window}
A substring $S[i.. j]$ of $S$ is called a
\emph{minimal unique palindromic substring~(MUPS)} of $S$ if and only if
$S[i.. j]$ is a palindrome, $S[i.. j]$ is unique in $S$, and
$S[i+1.. j-1]$ is repeating in $S$.
We denote $\MUPS(S)$ the set of intervals corresponding to MUPSs of $S$,
i.e., $\MUPS(S) = \{[i, j] \mid S[i..j] \mbox{ is a MUPS of }S\}$.
For example,
palindromic substring $S[9.. 13] = \mathtt{bbabb}$ of string $S = \mathtt{aaababababbabb}$
is a MUPS of $S$ since $S[9.. 13] = \mathtt{bbabb}$ is unique in $S$
and $S[10..12] = \mathtt{bab}$ is repeating in $S$.

Now, we show Lemma~\ref{lem:MUPS_eertree} which states a relationship between eertrees and MUPSs.
Then, in Lemma~\ref{lem:offline_MUPS}, we show that all MUPS can be computed using eertrees in an offline manner.
\begin{lemma} \label{lem:MUPS_eertree}
  A string $w$ is a MUPS of $S$ if and only if
  there is a node $v$ in $\eertree(S)$ such that
  $\pal(v) = w$, $\pal(v)$ is unique in $S$ and $\pal(u)$ is repeating in $S$,
  where $u$ is the parent of $v$.
\end{lemma}
\begin{proof}
  \noindent ($\Rightarrow$)
  Since $w$ is a MUPS of $S$,
  it is clear that there is a node $v$ such that
  $\pal(v) = w$ and it is unique in $S$.
  Also, since $\pal(v) = w \ne \varepsilon$, $v$ has the parent $u$,
  which represents the string $w[1..|w|-2]$.
  By the definition of MUPS, $\pal(u) = w[1..|w|-2]$ is repeating in $S$.
  \noindent ($\Leftarrow$)
  Since the palindrome $\pal(v) = w$ is unique in $S$
  and $\pal(u) = w[1..|w|-2]$ is repeating in $S$, $w$ is a MUPS of $S$.
\end{proof}
\begin{lemma}\label{lem:offline_MUPS}
  Given $\eertree(S)$,
  we can compute $\MUPS(S)$ in $O(\DPal(S))$ time.
\end{lemma}
\begin{proof}
  Given a node $v$, we can detect whether $\pal(v)$ is a MUPS or not in constant time by Lemma~\ref{lem:MUPS_eertree}.
  Also, the starting position of a palindrome $\pal(v)$ which is unique in $S$ is stored in $v.\bp.\first$.
  Therefore, we can compute $\MUPS(S)$ by a single traversal on $\eertree(S)$.
\end{proof}
Moreover, we can efficiently maintain MUPSs for a sliding window.
\begin{theorem} \label{thm:MUPS}
  We can maintain the set of MUPSs for a sliding window
  in a total of $O(n\log\sigma')$ time using $O(d)$ space.
\end{theorem}
\begin{proof}
  In addition to the eertree data structure described in Section~\ref{sec:eertree},
  we add 1-bit information $\ismups$ into each node.
  This bit $\ismups$ is set to $\mathtt{1}$
  if the node corresponds to a MUPS and to $\mathtt{0}$ otherwise.
  We first consider to delete the leftmost character $S[i-1]$ from $S[i-1.. j-1]$.
  In this case, only prefixes of $S[i-1.. j-1]$ are those
  whose number of occurrences in the sliding window change.
  We check the nodes corresponding to the longest and the second longest
  palindromic prefixes, and update $\ismups$ of them accordingly.
  We do not need to care about other palindromic prefixes
  since they must be repeating in $S[i.. j-1]$.
  Symmetrically, we can easily maintain $\ismups$ in the case
  when appending a character $S[j]$ to $S[i..j-1]$.
\end{proof}

\subsection{Computing Minimal Absent Palindromic Words for a Sliding Window}
A string $w$ is called a \emph{minimal absent palindromic word}~(\emph{MAPW})
of string $S$ if and only if
$w$ is a palindrome, $w$ does not occur in $S$, and
$w[1..|w|-2]$ occurs in $S$.
For example, palindrome $w = \mathtt{aabbaa}$ is a MAPW
of string $S = \mathtt{aaababababbabb}$
since $w$ does not occur in $S$ and 
$w[1..|w|-2] = \mathtt{abba}$ occurs in $S$ at position $8$.
For a relation between MAPWs and eertrees, the next lemma holds.
\begin{lemma} \label{lem:MAPW}
  For any non-empty string $w \in \Sigma^*$,
  $w$ is a MAPW of a string $S$ if and only if
  there is a node $u$ in $\eertree(S)$
  such that
  $\pal(u) = w[1..|w|-2]$, $\len(u) = |w|-2$, and $u$ does not have an edge labeled by $w[0]$.
\end{lemma}
\begin{proof}
  \noindent
  ($\Rightarrow$)
  Since $w$ is a MAPW, $w[1..|w|-2]$ is a palindromic substring of $S$,
  and thus, there is a node $u$ with $\pal(u) = w[1..|w|-2]$ and $\len(u) = |w|-2$.
  Also, since the palindrome $w$ does not occur in $S$,
  $u$ does not have an edge labeled by $w[0]$.
  \noindent
  ($\Leftarrow$)
  Since $u$ is a node in $\eertree(S)$, the string $\pal(u) = w[1..|w|-2]$ occurs in $S$.
  Also, since $u$ does not have an edge labeled by $w[0]$,
  the string $w$ does not occur in $S$.
  Thus, $w$ is a MAPW of $S$.
\end{proof}

In order to maintain the set of MAPWs on top of $\eertree(S)$,
we store an array $M_v$ of size $\sigma$ for each node $v$ in $\eertree(S)$
where $M_v[c] = \mathtt{0}$ if $v$ has an edge labeled by $c$
and $M_v[c] = \mathtt{1}$ otherwise.
By Lemma~\ref{lem:MAPW},
$M_v[c] = \mathtt{1}$ iff $c\pal(v)c$ is a MAPW of $S$.
It is easy to see that $M_v$ for all nodes $v$~(i.e., all MAPWs of $S$) can be computed
by traversing $\eertree(S)$ only once.
Thus, the next corollary holds.
\begin{corollary}
  The number of MAPWs of $S$ is at most $(|\DPal(S)|+1)\sigma$.
  Also, given $\eertree(S)$, the set of MAPWs of $S$ can be computed
  in $O(|\DPal(S)|\sigma)$ time.
\end{corollary}
Also, we can maintain MAPWs for a sliding window by applying Theorem~\ref{thm:eertree_nsigma}.
\begin{corollary}
  We can maintain the set of MAPWs for a sliding window
  in a total of $O(n + d\sigma)$ time using $O(d\sigma)$ space.
\end{corollary}
 
\end{document}